\documentclass[11pt,a4paper]{article}
\usepackage{amsthm, amsmath, amssymb}
\usepackage{mdwlist}
\usepackage{comment}
\usepackage{xspace}
\usepackage{color}
\usepackage{graphicx}
\usepackage{caption}
\usepackage{subcaption}
\usepackage{url}

\usepackage[margin=1in]{geometry}

\newcommand{\spread}{\Delta}
\newcommand{\aspect}{\mathrm{aspect}}

\newcommand{\e}{\varepsilon}

\newcommand{\R}{\mathbb{R}}

\newcommand{\ir}{\lambda}

\newcommand{\dist}{\mathbf{d}}

\newcommand{\lfs}{\mathbf{f}}

\newcommand{\vor}{\mathrm{Vor}}
\newcommand{\del}{\mathrm{Del}}

\newcommand{\ball}{\mathrm{ball}}

\newcommand{\vol}{\mathrm{vol}}

\newcommand{\etaCAGE}{\eta\mathrm{-CAGE}}

\newcommand{\xnbor}{\mathcal{N}}

\DeclareMathOperator*{\argmin}{argmin}
\DeclareMathOperator*{\argmax}{argmax}

  \newcommand{\pseudocode}[1]{\textbf{#1}\xspace}
  \newcommand{\If}{\pseudocode{if}}
  \newcommand{\Then}{\pseudocode{then}}
  \newcommand{\Else}{\pseudocode{else}}
  \newcommand{\ElseIf}{\pseudocode{else if}}
  
  \newcommand{\ForEach}{\pseudocode{for each}}
  \newcommand{\For}{\pseudocode{for}}
  \newcommand{\To}{\pseudocode{to}}
  \newcommand{\While}{\pseudocode{while}}
  \newcommand{\Return}{\pseudocode{return}}  
  \newcommand{\tab}{\indent}

  \newcommand{\algoline}{%
    \noindent
    \rule{\linewidth}{.25pt}
  }

    \newcommand{\algo}[1]{\textsc{#1}\xspace}
    \newcommand{\WellSpacedPoints}{\algo{WellSpacedPoints}}
    \newcommand{\Insert}{\algo{Insert}}
    \newcommand{\RegularInsert}{\algo{RegularInsert}}
    \newcommand{\Snap}{\algo{Snap}}
    \newcommand{\NewLayer}{\algo{NewLayer}}
    \newcommand{\FindNN}{\algo{GreedyWalk}}
    \newcommand{\ApproxFarCorner}{\algo{ApproximateFarthestCorner}}
    \newcommand{\UpdateAspect}{\algo{UpdateAspect}}
    \newcommand{\PruneEdges}{\algo{PruneEdges}}
    \newcommand{\Delete}{\algo{Delete}}
    \newcommand{\Refine}{\algo{Refine}}
    \newcommand{\InsertEdge}{\algo{InsertEdge}}
    \newcommand{\DeleteEdge}{\algo{DeleteEdge}}
    \newcommand{\InsertVertex}{\algo{InsertVertex}}
    \newcommand{\DeleteVertex}{\algo{DeleteVertex}}
    \newcommand{\Nbr}{\algo{Nbr}}
    \newcommand{\ApproximateGreedyPermutation}{\algo{ApproximateGreedyPermutation}}
    \newcommand{\Qrefine}{Q_{\mathrm{refine}}}

    \newcommand{\NN}{\algo{NN}}
    
\newtheorem{theorem}{Theorem}
\newtheorem{lemma}[theorem]{Lemma}
\newtheorem{corollary}{Corollary}

\newcommand{\fullversion}[1]{#1}
\newcommand{\shortversion}[1]{}

\title{A Fast Algorithm for Well-Spaced Points and Approximate Delaunay Graphs}

\author{
  Gary L.~Miller\thanks{Department of Computer Science,
    Carnegie Mellon University, {\tt glmiller@cs.cmu.edu} Partially supported by the NSF grant CCF-1065106.}
    \and
  Donald R.~Sheehy\thanks{Inria Saclay \^{I}le-de-France, {\tt don.r.sheehy@gmail.com} Partially supported by the European project No.~255827 (CG-Learning)} 
    \and
  Ameya Velingker\thanks{Department of Computer Science,
    Carnegie Mellon University, {\tt avelingk@cs.cmu.edu} Partially supported by the NSF grant CCF-1065106.}
}

\begin{document}

  \maketitle
  \begin{abstract}
We present a new algorithm that produces a well-spaced superset of points conforming to a given input set in any dimension with guaranteed optimal output size.
We also provide an approximate Delaunay graph on the output points. 
Our algorithm runs in expected time $O(2^{O(d)}(n\log n + m))$, where $n$ is the input size, $m$ is the output point set size, and $d$ is the ambient dimension. 
The constants only depend on the desired element quality bounds.

To gain this new efficiency, the algorithm approximately maintains the Voronoi diagram of the
current set of points by storing a superset of the Delaunay neighbors of each point. By retaining
quality of the Voronoi diagram and avoiding the storage of the full Voronoi diagram, a simple
exponential dependence on $d$ is obtained in the running time. Thus, if one only wants the
approximate neighbors structure of a refined Delaunay mesh conforming to a set of input points, the
algorithm will return a size $2^{O(d)}m$ graph in $2^{O(d)}(n\log n + m)$ expected
time.
If $m$ is superlinear in $n$, then we can produce a hierarchically well-spaced superset of size $2^{O(d)}n$ in $2^{O(d)}n\log n$ expected time.

\end{abstract}


  \section{Introduction} 
\label{sec:intro}

  Delaunay meshes are used extensively in graphics and scientific computing, and more recently, they have been applied in higher dimensions for data analysis~\cite{cheng12delaunay, hudson10topological}.
  However, there are serious difficulties in constructing meshes in more than three dimensions.
  Perhaps the two most pressing problems are the combinatorial blowup in size complexity and the numerical instability of the predicates used in the construction.
  We present a new algorithm that alleviates both of these problems.
  It works by only storing an approximation to the Delaunay graph.
  This has the simultaneous effect of eliminating much of the complexity of the mesh itself while at the same time eliminating the costly and unstable predicate computations.
  It also introduces a new walk-based point location method that can locate the nearest neighbor of a point to be inserted in constant time after $2^{O(d)}n\log n$ time for preprocessing.
  
  The specific class of meshes we consider arises from Delaunay or Voronoi refinement.
  The input is a set of $n$ points $P\subset \R^d$ and the output is a \emph{well-spaced} superset $M$ of size $m$.
  Delaunay refinement meshes have an extensive theory providing guarantees including optimal runtime, optimal output size, conformity to input boundaries, and even the removal of a class of poorly-shaped simplices called \emph{slivers}.
  These methods are widely used in practice and form the basis for several popular meshing codes such as Triangle~\cite{shewchuk96triangle}, TetGen~\cite{siTetGen}, and the CGAL 3D Mesh Generation package~\cite{cgalMesh3D}.
  Still, these codes have not been extended to higher than three dimensions.
  
  Delaunay refinement is inherently incremental.
  At every point in time, the data structure represents a valid Delaunay triangulation.
  This makes several things easier, but the convenience comes at a cost.
  The numerical computation of linear predicates often gives unstable solutions requiring exact arithmetic or other machinations to resolve.
  A significant amount of research has gone into more robust and efficient predicates and yet, for Delaunay refinement, only a small fraction of these computations will determine the structure of the output triangulation.
  However, the intermediate state must be exact to guarantee that the output is correct.
  Similarly, representing all of the simplices is a waste as only a fraction of these simplices appear in the output.
  
  We propose to perform a kind of Voronoi refinement that stores only a sparse graph as its intermediate state.
  This graph is an approximation to the edges of the Delaunay triangulation.
  Two new challenges arise when working with the approximate Delaunay graph.
  The first challenge is that the circumcenters of Delaunay simplices are usually used to refine the mesh.
  Without simplices, one must look elsewhere for good places to insert Steiner points.
  The second challenge is that the intermediate meshes are often used as point location data structures for the input points that have not been inserted yet. 
  Without the clear decomposition of space given by the Delaunay triangulation or the Voronoi diagram, one must find some other way to organize the uninserted points geometrically.
  
  For Steiner points, we show how approximate linear programming can give approximate Delaunay circumcenters without ever finding a single Delaunay simplex.
  For point location, we show how preprocessing the points in a special ordering allows for a constant time per point scheme.
  Thus, after $2^{O(d)}n\log n$ preprocessing time, the algorithm will run in output-sensitive linear time.
  
  
\paragraph{Our Contribution} 
  We present a new algorithm that computes an asymptotically optimal well-spaced superset of a set of points $P$.
  It has all the quality guarantees achieved by Delaunay or Voronoi refinement, but never actually constructs the full Delaunay triangulation or Voronoi diagram.
  The expected running time is $2^{O(d)}(n\log n + m)$, and the output graph has size $2^{O(d)}m$.
  If a hierarchically well-spaced point set is sufficient, then the expected running time is $2^{O(d)}n\log n$, and the output size $2^{O(d)}n$.
  

\paragraph{Related work} 

  Boissonat et al.~\cite{boissonnat09incremental} proposed to store only the edges of the Delaunay triangulation as a way to make the construction more tractable in higher dimensions.
  Still, their approach requires enumerating the higher dimensional simplices locally when making updates and thus still requires some high dimensional predicates.
  In order to avoid this work, it seems necessary to relax the condition that the exact Delaunay graph be stored.
    

  Well-spaced points are a standard method in mesh generation.
  The use of Delaunay refinement to produce well-spaced points dates back to Chew~\cite{chew89guaranteed}.
  Ruppert~\cite{ruppert95delaunay} developed the optimality theory for such meshes and the generalization of his quality criteria was formalized by Miller et al.~\cite{miller99radius}.
  The recent book by Cheng et al.~\cite{cheng12delaunay} gives a thorough treatment of Delaunay meshing in $\R^3$.

  The Sparse Voronoi Refinement algorithm of Hudson et al.~\cite{hudson06sparse} changed the perspective from Delaunay refinement to Voronoi refinement and was able to prove the first near-optimal output-sensitive algorithm with a running time of $O_d(n\log \spread + m)$.
  Here, $O_d$ hides constant factors that depend only on $d$.
  The term depending on the \emph{spread} $\spread$ was removed by Miller et al.~\cite{miller11beating}, who devised the NetMesh algorithm, an optimal comparison-based algorithm running in $O_d(n\log n + m)$ time.
  For $d$-dimensional quality meshes that include all simplices, these constants are necessarily superexponential in $d$.
  We avoid these superexponential constants by only storing a graph.



  \section{Background} 

  \paragraph{Geometry Basics} 
    We assume throughout that all points lie in $d$-dimensional Euclidean space.
    We identify the points of this space with vectors in $\R^d$ and use $\dist(u,v)$ to denote the Euclidean distance between $u$ and $v$.
    Similarly, we write $ \dist(u,S)$ for a point $u$ and a compact set $S$ to denote the minimum distance from $u$ to a point of $S$.  
    The \textbf{nearest neighbor} of $u$ in $S$ is 
    \[
      \NN(u,S) := \argmin_{v\in S\setminus\{u\}}  \dist(u,v).
    \]
    If the nearest neighbor is not unique, we choose $\NN(u,S)$ arbitrarily among the possibilities.
    When the underlying set $S$ is clear from context, we will often simply write $\NN(u)$.
    

  \paragraph{Voronoi Diagrams and Delaunay Triangulations} 

    The \textbf{Voronoi diagram}, $\vor_P$, of a set of points $P$ in $\R^d$ is a decomposition of $\R^d$ into polyhedral cells called Voronoi cells, where each cell is the set of reverse nearest neighbors of a point of $P$.  That is,
      $\vor_P(p) = \{x\in \R^d :  \dist(x,p) =  \dist(x,P)\}$. 
    These cells are closed and form a cell complex, i.e.\ the intersection of two cells is another polytope of one less dimension.
    The cells of all dimensions of the Voronoi diagram are called the \textbf{faces}.
    The dual of the Voronoi diagram is the \textbf{Delaunay triangulation}, $\del_P$.
    It has a $k$-simplex for every $(d-k)$-face of the Voronoi diagram.
    Moreover, the center of the ball circumscribing any $d$-simplex of $\del_P$ is a vertex of $\vor_P$. 

    The total number of faces of the Voronoi diagram, or dually, the number of simplices of the Delaunay triangulation can be as large as $\Theta(n^{\lceil d/2\rceil})$.
    However, if the points are more evenly spaced, the complexity drops to $\Theta(n)$.
    One view of Voronoi- or Delaunay-based mesh generation algorithms is that they seek to find such an arrangement of points for which this complexity is as small as possible.
    However, even in such instances, the dependence on the dimension is quite bad.
    In particular, the complexity is dominated by the higher dimensional simplices.
    This motivates our approach of storing only the edges of $\del_P$.

  
  \paragraph{Mesh Generation and Well-Spaced Points} 
    The \textbf{inradius} of a Voronoi cell $\vor_P(u)$ is the radius of the \textbf{inball}, the largest ball centered at $u$ that is contained in $\vor_P(u)$.
    It is denoted $r(u)$ and is equal to half the distance from $u$ to its nearest neighbor in $P$.
    The \textbf{outradius} of $\vor_P(u)$, denoted $R(u)$, is the radius of the smallest ball that contains all of the vertices of $\vor_P(u)$.
    The aspect ratio of $\vor_P(u)$ is defined as $\aspect_P(u) = R(u)/r(u)$.
    We say that $M$ is \textbf{$\tau$-well-spaced} if for all $u\in M$, $\aspect_M(u) \le \tau$.
    In such cases, the constant $\tau$ is sometimes called the \textbf{quality} of the mesh $\del_M$.
    
    For some point sets $P$, any well-spaced superset $M$ of $P$ will be prohibitively large.
    To combat this problem, Miller et al.\ defined \textbf{hierarchically well-spaced} points~\cite{miller11beating}.
    A set $M$ is hierarchically $\tau$-well-spaced if there is a rooted tree $T$ such that following hold.
    \begin{enumerate*}
      \item The nodes $V$ of $T$ are subsets of $M$.
      \item Each set in $V$ is $\tau$-well-spaced. 
      \item For any two $A,B\in V$, $A\cap B$ is at most a single point.
      If $A$ is the parent of $B$ then $A\cap B$ is a single point.
      If $A\cap B$ is a point then that point is also in every set in the path from $A$ to $B$ in $T$.
      \item If $A$ is the parent of $B$ in $T$ and $p = A\cap B$ then $B\subset \ball(p, \e\dist(v,A\setminus \{p\}))$ for some small constant $\e$ that will depend on the details of the algorithm.
    \end{enumerate*}
    Hierarchically well-spaced points provide a Delaunay refinement analogue of compressed quadtrees.
    Most importantly, for any $P\subset \R^d$, there exists a hierarchically well-spaced superset of size $2^{O(d)}n$~\cite{miller11beating}.
    For a wide class of reasonable inputs, the hierarchy is not necessary, in which case the output hierarchy will have a single set of linear size.
    However, for some inputs, no linear-size (non-hierarchical) well-spaced superset is possible~\cite{sheehy12new}.
    

  \paragraph{The Feature Size} 
  \label{par:the_feature_size}
  
    For any set $S$, $\lfs_S$ maps a point $x\in \R^d$ to the distance to the second nearest point to $x$ in $S$.
    For example, if $x\in S$ then $\lfs_S(x) = \dist(x, \NN(x))$ since $x$ itself is the nearest point in $S$.
    This $1$-Lipschitz function is known as the \textbf{feature size} with respect to $S$.
    Many standard guarantees of meshing algorithms are derived from looking at the feature size function.
    For example, to show that an algorithm terminates with an optimal size output, it suffices to show that $\lfs_P(v)\le K \lfs_M(v)$ for each $v\in M$, where $P$ is the input and $M$ is the output.


  \paragraph{Meshing Guarantees} 
  
  
    There are several desirable guarantees for an algorithm that produces a well-spaced superset $M$ of an input set $P$:
    \begin{itemize*}
      \item \textbf{Quality Output:} $\aspect_M(v)\le \tau$ for all $v\in M$.
      \item \textbf{Local Sizing:} $\lfs_P(v)\le K \lfs_M(v)$ for all $v\in M$ and a constant $K$ that depends only on the quality $\tau$.
      \item \textbf{Size Optimality:} $|M| = 2^{O(d)}|M'|$ for any other $\tau$-well-spaced superset $M'$ of $P$.
      \item \textbf{Running Time:} $O_d(n\log n + |M|)$.
    \end{itemize*}
    In this paper, we give an algorithm that achieves all of these guarantees.
    Moreover, the constants in the running time are at most $2^{O(d)}$.
  
  
  \paragraph{Delaunay and Voronoi Refinement} 
  
    The Delaunay refinement paradigm for mesh generation is an incremental algorithm that improves mesh quality by iteratively adding a Steiner point at the circumcenter of a simplex whose circumradius to shortest edge ratio is greater than some constant.
    The dual approach, known as Voronoi refinement, instead adds the vertex $u$ of $\vor(v)$ farthest from $v$ whenever $\vor(v)$ has aspect greater than some constant $\tau$.
    To avoid confusion with the input vertices, we call such a point $u$ the \textbf{farthest corner} of the Voronoi cell $\vor(v)$.
    Recall that vertices of the Voronoi diagram are circumcenters of Delaunay simplices, but the quality conditions are slightly different.
    Although asymptotically the same, Voronoi refinement will produce stronger guarantees on the degree of the Delaunay graph and the edge length discrepancy of edges incident to a vertex at the cost of doing more refinement.

    The key idea in both methods is to perform only local operations.
    Then, the structure of the current Delaunay triangulation or Voronoi diagram is used to find large empty regions to refine.
    
 
  \paragraph{Samples, Nets, and Cages} 
    A subset $S\subset P$ is an \textbf{$(\e,\delta)$-sample} of $P$ if all points of $P$ are within distance $\e$ of a point of $S$ and no pair of points of $S$ is separated by a distance less than $\delta$.
    Such samples are also known as Delone sets~\cite{clarkson06building}.
    The special case of $\e=\delta$ is sometimes referred to as a (metric) $\epsilon$-net.
    Note that such samples are not the same as well-spaced points as they are restricted to a subset and also cannot adapt their density locally as well-spaced points do. 
  
    In several places in the algorithm we need a net on the $(d-1)$-sphere $S_{d-1}$.
    We will refer to these special nets as \emph{cages} to avoid confusion with range space nets.
    We state the following folklore result (see Matousek~\cite{matousek02lectures} for a proof).
    \begin{lemma}\label{lem:nets_exist}
      For any $\eta \in (0,1]$ there exists a subset of points $C \subset S_{d-1}$ called an \textbf{$\eta$-cage} such that for all $x \in S_{d-1}$, $ \dist(x,C) \leq \eta$ and $|C| \leq O((4/\eta)^d$).
    \end{lemma}

    We will assume that we have precomputed a fixed $\eta$-cage for sufficiently small $\eta$ and refer to it as $\etaCAGE$.
    Since only one such cage is needed, we assume it is hardcoded into the algorithm.



  \paragraph{Permutations and Greed} 

    A permutation on a set of points induces an order relation.
    For a set $S$ with an order relation $\prec$ and a point $v\in S$, the \textbf{predecessor radius} of $v$ with respect to $S$ is defined as
        $\ir(v,S) := \min_{u\prec v}  \dist(u,v)$.
    It is often convenient to express an ordering on a set by labeling its points.
    For example, we write a permutation on $P$ as $p_1,\ldots,p_n$.
    The \textbf{prefixes} of this ordering are the sets $P_i :=\{p_1,\ldots, p_i\}$. 
    A \textbf{predecessor pairing} $\phi$ maps each point $p_i$ to $\phi(p_i) = \NN(p_i,P_{i-1})$ for $i\ge 2$.
    
    The \textbf{greedy permutation} of a set of points $P$ is an ordering that starts with any point $p_1$ and then satisfies $p_i = \argmax_{p\in P} \dist(p,P_{i-1})$.
    This is also sometimes known as a farthest point sampling strategy.
    
    A good approximation to the greedy permutation as well as a predecessor pairing can be computed in $2^{O(d)}n\log n$ time using an algorithm of Har-Peled and Mendel~\cite{har-peled06fast}.
    It is approximate in the sense that each point is approximately the farthest next point up to a factor of $1-n^{-\alpha}$ for some $\alpha>0$.
    To simplify the exposition, we assume that we compute a true greedy permutation.
    An approximation will only affect some constant factors.
    
  
  \paragraph{Approximate Linear Programming} 
    
    Our algorithm will employ linear programming (LP) using, for instance,
    the ellipsoid method  to find large empty regions in the point set
    corresponding roughly to the farthest corners of poor quality
    cells.  The running time of the ellipsoid method for
    $d$ dimensions and $\ell$ constraints is
    $O(\ell\, \text{poly}(d) \log\frac{1}{\epsilon})$, where $\epsilon$ is the error
    bound~\cite{GroetschelLovaszSchrijver1993}.  For the linear programs we care about,
    it suffices to compute the answer to a fixed constant error.  Moreover, the
    number of constraints will always be $2^{O(d)}$.  Thus, the time
    required to solve such an LP is also $2^{O(d)}$.

    
  

  \section{The Algorithm} 
\label{sec:the_algorithm}

\subsection{Overview of the Algorithm} 
\label{sub:overview}

  The basic structure of the algorithm is as follows: First, one computes a greedy permutation of the input points using the algorithm of Har-Peled and Mendel~\cite{har-peled06fast}.
  The points are then inserted one at a time, yielding an incremental construction of a hierarchically well-spaced set of points. 
  Each time an input point is inserted near an existing Steiner vertex $q$, the vertex $q$ is removed in a process we call \textbf{snapping}.
  After each input point insertion, Steiner vertices are added as part of a refinement procedure to
  retain quality. 
  At each step in the algorithm, we maintain an approximate Delaunay graph as well as the inradius and approximate farthest corner of the inserted vertices.
  Like walk-based algorithms for Delaunay triangulation, there is no auxiliary data structure required for point location~\cite{green78computing,mucke96fast}, but using greedy permutations allows us to prove stronger guarantees.
  \fullversion{Pseudocode for these routines may be found in Appendix~\ref{sec:pseudocode}.}
  \shortversion{Pseudocode for all of these routines may be found in the full version~\cite{ARXIV_VERSION}.}


\subsection{Data Structures} 
\label{sub:data_structures}

  The algorithm must maintain a hierarchy of well-spaced points and a sparse graph on each point set.
  The hierarchy is represented as a tree $T$.
  The nodes of $T$ are identified with the point sets.
  Recall that in the definition of a hierarchically well-spaced point set, the individual sets may intersect at exactly one point, and when this happens, then one set is the ancestor of the other.
  
  The main data structure used in the algorithm is a sparse graph representing the \textbf{approximate Delaunay graph} of the point sets in the hierarchy.
  The neighbors of a point $p$ are denoted $\xnbor(p)$.
  The edges of the approximate Delaunay graph contain the true Delaunay edges of the points.
  The extra edges incident to a point $p$ are all \emph{nearly} Delaunay in that they are approximately the same length as the true Delaunay edges incident to $p$ among the points in the same set of the hierarchy.
  
  Rather than a single approximate Delaunay graph, we keep a hierarchy of them, but we treat it as a single disconnected graph.
  However, each point gives a unique handle into one of the components; it is the lowest node in the tree containing that point.
  Hence, there is no need to navigate this hierarchy during the algorithm as it is merely an organizational tool.

  Unlike the previous work on hierarchically well-spaced points, the hierarchy grows monotonically, so, there is no need to maintain the hierarchy under dynamic deletions~\cite{miller11beating}.
  As a post-process, the collection of approximate Delaunay graphs is combined to produce a single approximate Delaunay graph.
  
  Attached to each vertex of each set in the hierarchy is an inradius $r(p)$, an approximate outradius $R(p)$, and an approximate farthest corner.
  


\subsection{Pre-processing} 
\label{sub:preprocessing}
  The first step of the algorithm is to pre-process the input points by computing an \textbf{approximate greedy permutation} $p_1,\ldots ,p_n$ of $P$.
  This will be the order in which we insert the input points.
  Let $\phi$ be the corresponding predecessor pairing. 
  
  We start by placing the input points in a bounding domain. 
  This will ensure that we do not insert too many points, because otherwise, local refinement rules can generate Steiner points to fill all of $\R^d$.
  We will use a copy of $\etaCAGE$ as the bounding domain scaled to have a radius that is a constant times larger than the diameter of $P$.
  It is known that for a reasonable constants, this will contain the refinement (see~\cite[Ch. 3]{sheehy11mesh} for a detailed discussion).


\subsection{Point Location} 
  Point location is needed only to identify the nearest neighbor of a point we are trying to insert among those vertices that we have already inserted.
  For a point $p$, a \textbf{greedy walk} routine locates this nearest neighbor of $p$.
  We shoot a ray from the predecessor $\phi(p)$ to $p$ and keep track of the Voronoi cells that the ray passes through.
  Since we insert the points in the order given by the greedy permutation, $\phi(p)$ will have been inserted previously.
  Even though we don't store the full Voronoi diagram, the approximate Delaunay graph suffices to find the next Voronoi cell that the ray passes through and, thus, the next step in the walk.

  Let $v$ be the vertex of the current position of the walk.
  The walk begins at $v = \phi(p)$.
  For each step, compute the intersection of the line through $p$ and $\phi(p)$ with the bisecting hyperplane between $q$ and $v$ for each neighbor $q$ of $v$.
  Computing these intersection points only requires scalar products rather than the determinant predicates.
  Among those intersection points for which the ray is leaving $\vor(v)$, choose the intersection that is closest to $v$.
  This is the true intersection of the boundary of $\vor(v)$ with the ray and indicates the next cell
  in the walk. 
  The walk terminates when the intersection is not in $\overline{p\phi(p)}$, and the
  last vertex is $\NN(p)$. 
  As we show in Lemma~\ref{lem:greedy_walk}, each greedy walk only takes a constant number of steps.

  The intersection of the halfspaces bounded by orthogonal bisectors of the Delaunay edges incident to a vertex is precisely the Voronoi cell of the vertex.
  Consequently, the correctness of this algorithm only requires that the approximate Delaunay graph contains the true Delaunay graph as a subgraph.
  The extra edges cannot change the walk because their orthogonal bisector hyperplanes all lie strictly outside the Voronoi cell.

\subsection{Incremental Updates}  
  The insertion routine adds a new point $p$ to the current data structure.
  It first finds the nearest neighbor of $p$ among those points inserted thus far using the greedy walk procedure described above.
  We then find the neighbors of $p$ in the approximate Delaunay graph by searching locally in the graph starting from $\NN(p)$.
  For input points, insertion may cause a nearby Steiner point to be deleted or it may cause the creation of a new layer in the hierarchy.
  These steps are described in more detail below.

  \paragraph{Inserting Input Points}  
  Let $p$ be an input point to be inserted.
  Once we have performed point location and determined $q = \NN(p)$, we perform one of three actions,
  depending on $q$ and the local sizing constant $K$. 
  \begin{itemize*}
    \item If $q$ is another input point and $\dist(p,q) \ge \frac{1}{K}r(q)$, we insert $p$. 
    \item If $q$ is another input point and $\dist(p,q) < \frac{1}{K}r(q)$, we add a new layer to the hierarchy (see below) and add $p$ to the new layer. 
    \item If $q$ is a Steiner point that has been added during a refinement stage, we \emph{snap} $q$ to $p$. 
    That is, we delete $q$ and insert $p$. 
  \end{itemize*}  
  In all cases, we finish by updating the approximate Delaunay graph as described below.

  \paragraph{Adding a New Set to the Hierarchy} 
    A new set of the hierarchy is always added as a new leaf to the hierarchy tree.
    This happens only when attempting to add an input point $p$ that is too close to another input point $q$ compared to the feature size at $q$.
    In such cases, we replace the point $q$ with a copy of itself and move $q$ to the newly created set.
    We then add a new copy of $\etaCAGE$ centered at $q$ with radius $2\,\dist(p,q)$ and add $p$ and $q$ to the new level.
    This is a new connected component of the approximate Delaunay graph and starts as a complete graph.
    This roughly corresponds to restarting the algorithm with just the points $p$ and $q$, except there is no need to recompute the greedy permutation.

    If an input point has nearest predecessor $q$, then this refers to the copy of $q$ in the new set in the hierarchy.
    This convention makes it possible to ignore the hierarchy when doing point location as the greedy permutation guarantees that any changes near $q$ will happen at the smaller scale and thus in the newly created layer.

  \paragraph{Finding Approximate Delaunay Neighbors}  
    The insertion of $p$, along with the possible removal of $q = \NN(p)$ in the case of snapping, alters the current approximate Delaunay graph $G$. 
    To find the neighbors of $p$ in the updated graph, we do a breadth-first search from $q$ in $G$.
    The search is pruned at any vertex $v$ that cannot be a neighbor of $p$ because either $\dist(p,v)> 2R(v)$ or $\dist(p,v) > 2\tau'' r(p)$, where $\tau''$ is an upper bound on the aspect ratio of any vertex at any time in the algorithm.
    The number of neighbors found by this search is at most $2^{O(d)}$ (Lemma~\ref{lem:nrb_bound}).
    They are used to estimate the outradius $R(p)$ after inserting $p$ (as described below).
    The list of new neighbors is then limited to those vertices $v$ found in the search such that $\dist(p,v)\leq 2R(p)$, as $v$ cannot otherwise be a Delaunay neighbor of $p$.
    Only the neighbors of $p$ need to be updated by recomputing inradii, outradii, and approximate farthest corners.
    
    When deleting a point $p$, we update the approximate Delaunay graph by adding edges between all pairs $(u,v)$ of neighbors of $p$.
    Then, we recompute the approximate aspect ratios of the points and eliminate any edges for which $\dist(u,v)> 2\min\{R(u),R(v)\}$.

\subsection{Approximate Farthest Corner} 
  For each point $p$ whose approximate Delaunay neighborhood is updated, we update the approximate farthest corner of $\vor(p)$ as follows.

  Given a unit vector $c\in S^{d-1}$, we can find the corner of the Voronoi cell of $p$ that is farthest in the direction of $c$ by solving a linear program.
  The Voronoi cell of $p$ is precisely the set of points satisfying a linear constraint for each $q\in\xnbor(p)$ given by the hyperplane bisecting $\overline{pq}$. 
  Thus, the LP with these constraints and objective function given by the inner product with $c$ has a solution on the boundary of $\vor(p)$ that is extremal in the direction of $c$.

  If $c$ is chosen to be sufficiently close to the unit vector in the direction of the farthest corner $v$, then $v$ maximizes the objective function of our linear program. 
  However, the problem is that we do not have any \emph{a priori} knowledge of the optimal vector $c$. 
  Thus, we will solve the aforementioned problem for several choices of $c$. 
  In particular, we iterate $c$ over all points in $\etaCAGE$ in order to guarantee that some $c$ is close to the direction of $v$.

  To solve the LP, we can either use an interior point or ellipsoid method, as both would run in polynomial time in $d$ and the number of constraints. 
  In our case, the dimension is quite small while the number of constraints is exponential in the dimension.
  Thus, asymptotically we get better bounds using the ellipsoid method.
  Moreover, the aspect ratio bounds make it easy to find a sufficiently good starting ellipse. 
  We discuss a scheme based on the ellipsoid method for our problem.

  In the ellipsoid method, we need an ellipse that contains the Voronoi region, but in our application the aspect ratio of the Voronoi cells is at most a constant $\tau''$. Thus, for a given point $p$, $\ball(p, r(p)\tau'')$ contains the feasible region. Moreover, we shall only require an approximate solution, namely, a point $x$ for which $c^T x$ is within $(1-\epsilon)$ of the optimum. Thus, for some value of $t$, we add the additional constraint $c^T x \geq t$ to our LP. We then search for a feasible solution. The trick is to perform a binary search on $t\in [r(p), r(p)\tau'']$ and solve the LP feasibility for different values of $t$. Since we desire only an approximate solution within a $(1-\epsilon)$ factor, we obtain a lower bound on the volume of the feasible region and can bound the number of iterations of the ellipsoid method required.
  
  Each LP feasibility computation will require time $O(m+\mathrm{poly}(d))\log(\tau''/\epsilon)$, where $m$ is the number of constraints (in our case, $2^{O(d)}$). The binary search will add an additional multiplicative factor of $\log(\tau''/\epsilon)$. Moreover, repeating the procedure for all $c\in\etaCAGE$ will incur an additional $2^{O(d)}$ factor.
  
%

\subsection{Refinement} 
The insertion of a new input point may cause some Voronoi cells to have aspect ratio more than $\tau$. 
The refinement phase adds Steiner points in order to recover quality.
                                                                                          
Having computed the approximate farthest corner for all updated vertices, we have an approximate aspect ratio for each Voronoi cell.
We wish to keep a queue of all Voronoi cells whose aspect ratio is greater than $\tau$. Since we only know \emph{approximate} aspect ratios, we instead add to the queue all Voronoi cells for which our computed approximate aspect ratio is greater than $\tau(1-\epsilon)\left(1-\frac{\eta^2}{2}\right)$. This is because any Voronoi cell with a true aspect ratio of more than $\tau$ must necessarily satisfy this criterion, as is shown in Corollary \ref{cor:approxfarthest}.

The refinement stage consists of going through the queue of bad aspect ratio cells and inserting the approximate farthest corner of each as a Steiner vertex.
If a vertex is added to the queue but some other refinement causes it to no longer have bad aspect ratio by the time we process it, then we simply do nothing and move on to the next cell in the queue.
Each insertion requires that we update the approximate Delaunay graph exactly as we did for input points.
The point location cost for Steiner vertices is constant time per vertex because the nearest neighbor of the Steiner point is just the vertex whose cell had bad aspect ratio in the first place.

Refinement can cause new cells to have bad aspect ratio.
These are also added to the queue after each insertion.
Refinement continues until the queue of Voronoi cells is empty.


\subsection{Post-processing} 
\label{sub:post_processing}

  The hierarchy can be reconciled into a connected approximate Delaunay graph as a post-process in one of two ways.
  A simple, output-sensitive linear-time algorithm is known for extending a hierarchically well-spaced point set into a single well-spaced set~\cite{miller11beating}.
  The same algorithm can be adapted for our algorithm as it only requires the ability to refine Voronoi cells.
  However, since this operation could potentially add a superlinear number of Steiner points, we instead opt to complete the approximate Delaunay graph by adding edges only.
  
  Let $A,B$ be two nodes of the hierarchy tree such that $A$ is the parent of $B$.
  Let $v$ be the vertex shared by both $A$ and $B$ and let $G_A$ and $G_B$ be the approximate Delaunay graphs of $A$ and $B$ respectively.
  Let $C$ be the vertices in the boundary cage of level $B$ and let $D$ be the vertices adjacent to $v$ in $G_B$.
  We replace the nodes $A$ and $B$ with a new node whose vertex set is $A\cup B$ and its edge set contains the union of the edges in $G_A$ and $G_B$ as well as the complete bipartite graph on $(C,D)$.
  The new tree is the formed from the old one by contracting the edge from $A$ to $B$.
  Continuing this way allows us to flatten the entire hierarchy into a single, connected approximate Delaunay graph.
  


  \section{Algorithmic Invariants} 
\label{sec:invariants}

  The analysis follows closely that of previous work on Voronoi refinement meshing~\cite{hudson06sparse, miller11beating}.
  We will show that certain geometric properties are satisfied after each point is inserted by the algorithm.
  The \textbf{Feature Size Invariant} (Theorem~\ref{thm:fs_invariant}) guarantees that no two output points are too close to each other compared to the feature size of the input points.
  It is used to prove the termination of the algorithm and the asymptotic optimality of the output size (see~\cite[Ch. 3]{sheehy11mesh} for a general version of this fact).
  
  The other invariant we maintain is the \textbf{Quality Invariant} (Theorem~\ref{thm:quality_invariant}).
  It guarantees that the points in each set of the hierarchy are well-spaced.
  This is a critical requirement for a point set to be hierarchically well-spaced.
  It also guarantees, among other things, that the approximate Delaunay graph has constant bounded degree throughout the algorithm.
  
  These invariants will apply to the whole hierarchy of point sets.
  The sets in the hierarchy do not interact, because the approximate Delaunay graph has no edges between points in different sets.
  As a result, it suffices to prove each invariant individually, for each set in the hierarchy.
  
  At each step of the algorithm, there is an ordering of the inserted points.
  This ordering is mostly determined by the order in which the points were added.
  The only exception is that if an input point $p$ causes a nearby Steiner point $q$ to snap, then $p$ is given the index of $q$ in the updated ordering.

  Let $M_1,\ldots, M_N$ be the sequence of ordered point sets for each step of the algorithm.
  Similarly, let $P_i = M_i\cap P$ for each $i=1\ldots N$.
  Each step either adds an input point, adds a Steiner point, or snaps a Steiner point to an input point. 
  If the $i$th step is one of the first two cases, then $M_i$ is formed from $M_{i-1}$ by appending a new point to the end of the ordering.
  If there is a snap then $M_i$ is formed from $M_{i-1}$ by swapping one Steiner point for an input point.
  It is clear that for all $i<j$, $P_i \subseteq P_j$.

  The reader familiar with other analyses of Delaunay or Voronoi refinement schemes will notice a difference here in how we define the steps of the algorithm.
  In previous work, one considered a single ordered point set and reasoned about the prefixes of that ordering.
  This is also standard in analyses of other incremental geometric algorithms.
  In our case, we may undo some Steiner point insertion if it was too close to an input point, so we need a more nuanced analysis.
  The payoff is that we avoid the preemptive point location for Steiner points used in sparse Voronoi refinement~\cite{hudson06sparse}, and we get by without any extra point location data structures.

  
\subsection{The Feature Size Invariant} 
\label{sec:fs_invariant}

We want to show that at each step in the algorithm, the feature size with respect to the inserted points is larger than some constant times the feature size with respect to the input points.
Recall that the $\ir(v,M)$ is the predecessor radius of the point $v$ in the ordered set $M$.
Usually, $\ir(v,M_i) = \ir(v,M_j)$ for any $M_i$ and $M_j$ containing $v$.
However, for some points, snapping a Steiner point could cause the predecessor radius to change.
If it goes up, that only makes the feature size invariant easier to satisfy.
If it goes down, we need to carefully check that it does not go down by too much.
The following standard lemma shows that bounding the predecessor radius suffices to bound the feature size.

\begin{lemma}\label{lem:predecessor_radius_bound_implies_fs_bound}
  Let $M\subset \R^d$ be an ordered set with order relation $\prec$.
  Let $f:\R^d\to \R$ be any $1$-Lipschitz function and let $K>0$ be any constant.
  If for all $v\in M$, $f(v)\le (K-1)\ir(v,M)$, 
  then for all $v\in V$, $f(v)\le K\lfs_M(v)$. 
\end{lemma}
\begin{proof}
  Fix any $v\in M$ and let $u= \NN(v,M)$.
  If $u\prec v$ then $\lfs_M(v) =  \dist(u,v) = \ir(v,M)$, and so,
    $f(v) 
      \le (K-1)\ir(v,M)
      = (K-1) \lfs_M(v)
      < K \lfs_M(v)$.
  If $v\prec u$, then $\ir(u,M)\le \dist(u,v)$, and so,
    $f(v) 
      \le f(u) +  \dist(u,v)
      \le (K-1)\ir(u,M) +  \dist(u,v)
      \le K \dist(u,v)
      \le K\lfs_M(v)$.
\end{proof}

\begin{theorem}[The Feature Size Invariant]\label{thm:fs_invariant}
  For all $i = 1\ldots N$, and for all $v\in M_i$, 
    $\lfs_{P_i}(v)\le K\lfs_{M_i}(v)$,
  where $K = \frac{4\tau}{\tau-4}$.
\end{theorem}
\begin{proof}
  Lemma~\ref{lem:predecessor_radius_bound_implies_fs_bound} implies that it will suffice to bound the predecessor radius of each point in each ordering.
  Define the subset $X_i\subset M_i$ to be the points for which some previous snap caused a decrease in the predecessor radius, i.e.\ 
    $X_i:= \{v\in M_i : \ir(v,M_j) < \ir(v, M_{j-1}) \text{ for some } j\le i\}$.

  We will proceed by induction.
  The inductive hypothesis is that for all $k < i$ and all $v\in M_k$,
  \[
    \lfs_{P_i}(v) \le \left\{
      \begin{array}{ll}
        K_P\ir(v, M_k) & \text{if $v\in P_i\cup X_k$}\\
        K_S\ir(v, M_k) & \text{if $v\in M_k\setminus P_i\setminus X_k$}\\        
      \end{array}
    \right.
  \]
  where $K_S = \frac{\tau+4}{\tau-4}$ and $K_P = \frac{3\tau+4}{\tau-4}$.
  Note, this choice of $K_P$ and $K_S$ implies
  \begin{equation}\label{eq:fs_constant_bounds}
    K_P = K-1, ~ K_S = \frac{2K}{\tau} +1, \text{ and }  K_P = 2 K_S + 1.
  \end{equation}
  
  We will now prove that this hypothesis also holds for $M_i$.
  There are three cases to consider, a new Steiner point, a new input point whose nearest neighbor is also an input point, or a Steiner point snapped to an input point.
  
  \paragraph{Case 1: Adding a Steiner point} 
    If we are adding a Steiner point $v$, then the inductive hypothesis continues to hold for all points in $M_i$ other than $v$.  
    We need only show that it holds for $v$ as well.
    Let $x$ be the vertex whose cell had bad aspect ratio.
    Since $v$ is in the Voronoi cell of $x$ before insertion, $\ir(v, M_i) =  \dist(x,v)$.
    So, we may assume $2 \dist(x,v)/\lfs_{M_{i-1}}(x)> \tau$.
    Moreover, by induction and Lemma~\ref{lem:predecessor_radius_bound_implies_fs_bound}, $\lfs_{P_i}(x) \le K \lfs_{M_{i-1}}(x) = K \lfs_{M_i}(x)$, where the last equality follows because $v$ is not the nearest neighbor of $x$ in $M_i$.
    So, we have that 
      $\lfs_{P_i}(x) < 2 K\dist(x,v)/\tau$.
    Using these facts and the Lipschitz property of $\lfs_{P_i}$, we derive the following.
    \begin{align*}
      \lfs_{P_i}(v) 
        &\le \lfs_{P_i}(x) +  \dist(x,v)\\
        &< \left(\frac{2K}{\tau} + 1\right) \dist(x,v)
        = K_S\ir(v,M_i).
    \end{align*}

  
  \paragraph{Case 2: Adding an input point (no snapping)} 
  Since the only change in the ordering is to append the new vertex $v\in P_i$ to the end of the ordering, the inductive hypothesis will continue to hold for all points other than $v$ in $M_i$.
  Let $u = \NN(v, M_i)$.
  Since we did not snap, it must be that $u\in P_i$.
  So, $\lfs_{P_i}(v) \le  \dist(u,v) = \ir(v,M_i) \le K_P \ir(v,M_i)$.
  So, the hypothesis holds for $v$ as well.
  
  
  \paragraph{Case 3: Snapping a Steiner point to an input point} 
  
    Let $v$ be the new point added and let $x$ be the Steiner point that was removed.
    It will suffice to check that the inductive hypothesis holds for $v$ as well as any points for which the predecessor radius has been decreased by the snap.
    For such a point $u$, it must be that $x$ precedes $u$ in $M_{i-1}$ and so using the triangle inequality and that $ \dist(x,v)\le  \dist(u,v)$,
    \begin{equation}\label{eq:ir_halves}
      \ir(u,M_{i-1}) \le  \dist(x,u) \le 2 \dist(u,v) = 2\ir(u,M_i).
    \end{equation}
    
    \paragraph{Case 3(a): $u=v$} 
    
      By the definition of $x$, we have the $ \dist(x,v) \le \ir(v, M_i)$.
      So, by the triangle inequality, 
        $\ir(x, M_{i-1}) \le \ir(v, M_i) +  \dist(x,v) \le 2 \ir(v, M_i)$.
      If $x\in X_{i-1}$, then there is some point $y\in P_i$ such that $\ir(x,M_{i-1}) =  \dist(x,y)$.
      Using \eqref{eq:ir_halves}, we get the following bound.
      \[
        \lfs_{P_i}(v) \le  \dist(v,y) \le  \dist(v,x) +  \dist(x,y) \le 3\ir(v, M_i) \le K_P\ir(v, M_i).
      \]      
      Otherwise, if $x\notin X_{i-1}$, then by induction and \eqref{eq:ir_halves},
        $\lfs_{P_i}(x) \le K_S\ir(x,M_{i-1}) \le 2 K_S \ir(v,M_i)$.
      So, we get that
        $\lfs_{P_i}(v) \le \lfs_{P_i}(x) +  \dist(x,v) \le (2 K_S + 1)\ir(v,M_i) = K_P\ir(v,M_i)$.
    
    
    \paragraph{Case 3(b): $u\in P_i\setminus\{v\}$} 
        $\lfs_{P_i}(u) \le  \dist(u,v) = \ir(u, M_i) \le K_P\ir(u,M_i)$.
    
    \paragraph{Case 3(c): $u\in X_{i-1}\setminus P_i$} 
      The predecessor radius is determined by an input point for all points of $X_i$, so $\lfs_{P_i}(u) \le \ir(u,M_{i-1})$.
      Using \eqref{eq:ir_halves}, we have
        $\lfs_{P_i}(u) \le \ir(u,M_{i-1}) \le 2\ir(u,M_i) \le K_P\ir(u,M_i)$.

    \paragraph{Case 3(d): $u\in X_i \setminus X_{i-1}\setminus P_i$.} 
      By induction and \eqref{eq:ir_halves}, we have
        $\lfs_{P_i}(u)\le K_S\ir(u,M_{i-1}) \le 2 K_S\ir(u,M_i) < K_P \ir(u,M_i)$.
\end{proof}

There is one special case of the preceding theorem that is sufficiently useful that we state it below as its own lemma.
Its proof is the same as above, but only requires the first case.

\begin{lemma}\label{lem:steiners_preserve_fs}
  Let $M_i,\ldots, M_k$ be a contiguous sequence of ordered sets produced by the algorithm in which each was formed from the previous one by adding a Steiner point.
  Then, for all $j = i\ldots k$, and for all $v\in M_j$, 
    $\lfs_{M_i}(v)\le K\lfs_{M_j}(v)$,
  where $K = \frac{4\tau}{\tau-4}$.
\end{lemma}

\subsection{The Quality Invariant} 
\label{sec:quality}

  We now show that at all times, i.e.\ for each $M_1,\ldots M_N$, the points in each set of the hierarchy are well-spaced for a quality constant that does not depend on the dimension.
  In particular, this means that each $M_i$ is indeed a hierarchically well-spaced point set.
  
  \begin{theorem}[The Quality Invariant]\label{thm:quality_invariant}
    There exists a constant $\tau''$ independent of $n$ and $d$, such that at all times, the set of inserted points is $\tau''$-well-spaced. 
  \end{theorem}

  The proof of this fact follows exactly the pattern of previous work~\cite{hudson06sparse,miller11beating} and is broken into Lemmas~\ref{lem:inputs_preserve_quality} and~\ref{lem:steiners_preserve_quality} below.
  First we show that inserting an input point can decrease the aspect ratio from $\tau$ to a constant $\tau'$.
  Then, during the refinement the aspect ratio never dips below $\tau''$.
  By construction, the refinement phase will end with aspect ratio at most $\tau$ as computed using the approximate farthest corner routine.
  Since this is approximation, the true aspect ratio of the cells may be worse by a small constant factor.

  \begin{lemma}\label{lem:inputs_preserve_quality}
    Adding one point to a $\tau$-well-spaced set using the \Insert routine results in
a $\tau'$-well-spaced set, where $\tau' = 2K\tau$.
  \end{lemma}
  \begin{proof}
    Let $p$ be the newly inserted point.
    The only way for the aspect ratio of a point $q$ to increase is if its inradius decreases.
    This can only happen if the nearest neighbor is $p$.
    If $p$ was added to the Voronoi cell of $q$, then $q$ must be an input point for otherwise, the algorithm would have snapped $q$ to $p$.
    Moreover, for $q\in P$, we only insert $p$ if $\dist(p,q) \ge r(q)/K$, so $r(q)$ decreases by at most a factor of $2K$.
    For any other $q$ whose Voronoi cell's inradius decreases, $p$ cannot be more than twice as close to $q$ than the previous nearest neighbor of $q$.
    In that case the aspect ratio of $q$ goes down by at most a factor of $2$.
  \end{proof}

  To prove that a sequence of Steiner points added during refinement maintains a bound on the quality, we will use the following lemma from Hudson et al.~\cite[Lemma~6.1]{hudson06sparse}.
  
  \begin{lemma}\label{lem:fs_bounded_in_gap_balls}
    Let $M$ be a set of $\tau'$-well-spaced points, let $c$ be a point in the convex closure of $M$, and let $B = \ball(c,r)$ be a ball that does not contain any points of $M$ in its interior.
    Then, for all $x\in B$, $\lfs_M(x)\ge \frac{r}{16\tau'^2}$.
  \end{lemma}
  

  \begin{lemma}\label{lem:steiners_preserve_quality}
    Let $M_i$ and $M_k$ be the points of a mesh before and after a sequence of Steiner points were added, where $M_i$ is $\tau'$-well-spaced points.
    Then, for each $j=i\ldots k$, $M_i$ is $\tau''$-well-spaced, where $\tau''= 32 K \tau'^2$.
  \end{lemma}
  \begin{proof}
    Fix any intermediate mesh $M_j$ and any vertex $v\in M_j$.
    Let $R = R(v)$ be the outradius of $\vor_{M_j}(v)$.
    So, the aspect ratio of $\vor_{M_j}(v)$ is bounded using Lemmas~\ref{lem:steiners_preserve_fs} and~\ref{lem:fs_bounded_in_gap_balls} as follows.
    \[
      \aspect_{M_j}(v) 
        = \frac{2R}{\lfs_{M_j}(v)}
        \le \frac{2KR}{\lfs_{M_i}(v)}
        \le 32K\tau'^2.
    \]
  \end{proof}
  
  \paragraph{A Note About Constants} 
  \label{par:a_note_about_constants_}
  
    This section is optimized for concise proofs rather than optimal constants. 
    The main slack that can be tightened by a more detailed analysis comes from the constant $\tau'$ defined in Lemma~\ref{lem:inputs_preserve_quality}.
    Although the constant is tight, at most two points can achieve the bound at any given time.
    In the application of Lemma~\ref{lem:fs_bounded_in_gap_balls}, it is assumed that all points may have aspect ratio $\tau'$, and this assumption leads to a much looser bound than could be realized.
  


  \section{Correctness and Running Time} 
\label{sec:correctness}

We now show that the described algorithm correctly computes a hierarchically well-spaced point set and maintains a
superset of the edges in the Delaunay graph in each set in the hierarchy.

\subsection{Delaunay Neighbors and Degrees}
\label{sec:delaunay_nbrs}
In this section, we prove the two most important properties of the approximate Delaunay graph, namely that
the graph contains the true Delaunay graph and also has bounded degree.
We will assume that the points are in sufficiently general position that the Delaunay triangulation is indeed a triangulation.
If this is not the case, any infinitesimal (or symbolic) perturbation will make it so.

\begin{theorem}\label{thm:approx_delaunay_contains_true_delaunay}
  The approximate Delaunay graph produced by our algorithm contains all edges of the Delaunay triangulation of each point set in the output hierarchy of well-spaced points.
\end{theorem}
\begin{proof}
  The proof is by induction on the number of insertions and deletions performed by the algorithm.
  We show that at each step, the approximate Delaunay graph contains all of the Delaunay edges of each set in the hierarchy.
  
  As a base case, note that the starting cage connecting all pairs of vertices clearly contains all Delaunay edges.
  Assume inductively that the set $M_{i-1}$ from the first $i-1$ insertions and deletions satisfies the inductive hypothesis.
  
  Let $M_i$ be the set formed by inserting $p$ into $M_{i-1}$.
  By the Quality Invariant (Theorem~\ref{thm:quality_invariant}), there are no Delaunay edges $(p,q)$ longer than $\min\{2R(p),2R(q), 2\tau'' r(p)\}$.
  So, when the algorithm prunes such edges from the approximate Delaunay graph, it does not remove any Delaunay edges.
  It will therefore suffice to prove that all new Delaunay edges formed by inserting $p$ are added to the approximate Delaunay graph.
  All such edges are necessarily incident to $p$.
  At least one such edge is the edge from $p$ to its nearest neighbor $q = \NN(p,M_{i-1})$.
  Note that the neighbors of $p$ in $\del_{M_i}$ induce a connected subgraph of the Delaunay triangulation.
  Indeed, by Delaunay/Voronoi duality, this graph is isomorphic to the edges of the polytope polar to $\vor_{M_i}(p)$, and polytopal graphs are connected.
  So, the graph search starting at $q$ will find all of the Delaunay neighbors of $p$.
  
  Now, we consider the case of a point deletion that occurs when we snap.
  Let $M_i$ be the set of points formed from removing a point $p$ from $M_{i-1}$.
  In this case, the only new Delaunay edges in $\del_{M_{i}}$ that are not in $\del_{M_{i-1}}$ are those that connect two neighbors of $p$.
  Since we tentatively add all such pairs before pruning those that cannot be Delaunay, the updated graph again contains all Delaunay edges.
\end{proof}
We now prove an upper bound on the degree of a node in the approximate Delaunay graph at any time in the algorithm.
\begin{lemma}\label{lem:nrb_bound}
  In the approximate Delaunay graph of a well-spaced mesh with quality $\tau''$, each point has
at most $(3\tau''^2)^d$ neighbors. 
\end{lemma}
\begin{proof}
Let $p$ any vertex in the approximate Delaunay graph. 
Let $\ell$ be its degree and let $q_1, \ldots, q_\ell$ be its neighbors. 
Consider the inballs of $q_1, \dots, q_\ell$.
As an immediate consequence of the pruning condition, each neighbor $q_i$ has inradius $r(q_i) \le \frac{1}{2}\dist(p,q_i) \le \tau''r(p)$. Similarly, $r(q_i) \geq r(p)/\tau''$.
Moreover,
\[
  \dist(p,q_i) \le 2\tau''r(p).
\]
Therefore, all the inballs lie within a ball of radius $\dist(p,q_i) + \max_i\{r(q_i)\} \le 3\tau''r(p)$ centered at $p$. Since the inballs are disjoint, the following volume packing argument establishes the claim.
\[
  \ell \leq \frac{\vol\left(\ball(3\tau''r(p))\right)}{\vol\left(\ball(r(p)/\tau'')\right)} = (3\tau''^2)^d.
\]
\end{proof}

\subsection{Calculating Approximate Aspect Ratios}
\label{sec:approx_aspect_ratio}
We now prove the correctness of our method of approximating the aspect ratio of a Voronoi cell by finding an approximate farthest corner.

First, we show that we can approximate the farthest corner of a Voronoi cell in any specified direction in the desired running time by using the ellipsoid method. For each neighbor $q$ of $p$, we create a linear constraint that $x\in\R^d$ lies on the same side of $\mathcal{H}(p,q)$ as $p$ (here, $\mathcal{H}(x,y)$ denotes the hyperplane bisecting $\overline{xy}$). Let $\mathcal{L}$ be the set of all such constraints. By Lemma~\ref{lem:nrb_bound}, there are $2^{O(d)}$ constraints. Moreover, it is clear that the feasible region of the linear program defined by these constraints is precisely $\vor(p)$.
\begin{theorem}\label{thm:ellipsoid}
 Let $p$ be a point in the mesh which is, without loss of generality, located at the origin, and let $c\in \R^d$ be a vector with $\|c\| = 1$. Suppose $w$ is
the farthest corner of $\vor(p)$ in the direction of $c$, i.e., $w$ is the point of
$\vor(p)$ that maximizes $w^T c$. Then, in $2^{O(d)}
O\left(\log^2\left(\frac{\tau''}{\epsilon}\right)\right)$ time, we can calculate a point
$z\in\vor(p)$ such that $z^T c \geq (1-\epsilon) w^T c$.
\end{theorem}
Before we prove the above theorem, we prove a lemma.
\fullversion{
\begin{figure*}[ht]
\centering
 \begin{subfigure}[t]{0.4\textwidth}
  \centering
  \includegraphics[height=2in]{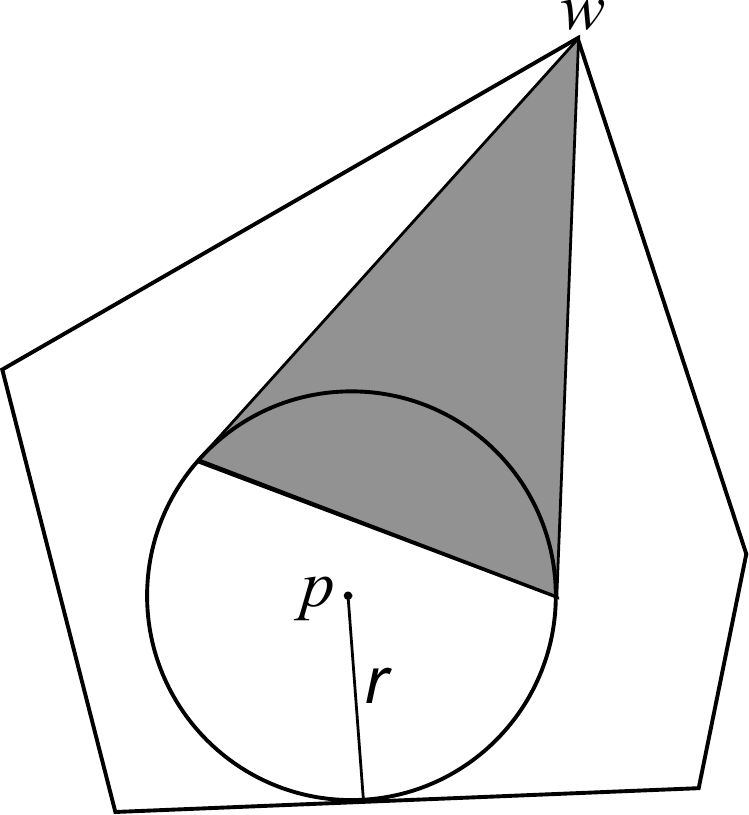}
  \caption{The shaded volume shows the cone formed by joining $w$ to $\ball(p,r)$. It is contained in
the Voronoi cell of $p$.}
  \label{fig:cone}
 \end{subfigure}
\quad\quad\quad\quad\quad
 \begin{subfigure}[t]{0.4\textwidth}
  \centering
  \includegraphics[height=2in]{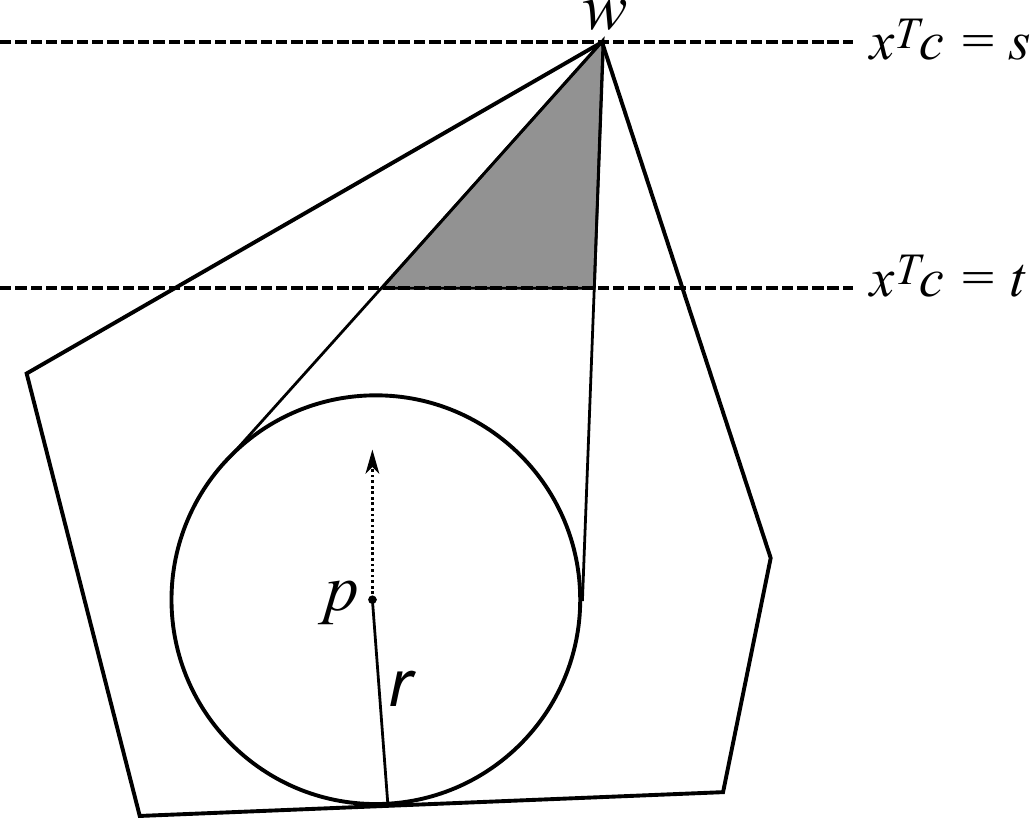}
  \caption{The shaded volume shows the cone formed by taking the intersection of the cone in (a)
with the half space $\{x\in\R^d: x^T c\geq t\}$. The dotted arrow indicates the direction given by
$c$.}
  \label{fig:cone_trunc}
 \end{subfigure}
 \caption{Bounding the volume of the feasible region during the ellipsoid method}
\end{figure*}
}

\begin{lemma}\label{lem:threshellipsoid}
 Let $p$, $c$, $w$, and $\mathcal{L}$ be as before, and assume $s = w^T c$. Write $r=r(p)$ and $R =
R(p)$. Suppose that $t \in [r, r\tau'']$ and $s \geq \frac{1}{1-\epsilon}r$. Then, in $2^{O(d)}
O\left(\log\left(\frac{\tau''}{\epsilon}\right)\right)$ time, we can produce output according to
the following rules:
 \begin{enumerate}
  \item If $t < s \left(1-\frac{\epsilon}{2}\right)$, then the output should be a point $x\in\R^d$ that satisfies the constraint set $\mathcal{L} \cup \{x^T c \geq t\}$.
  \item If $t > s$, then the output should indicate ``NOT FOUND.''
  \item If $s\left(1-\frac{\epsilon}{2}\right) \leq t \leq s$, then the output can be either a point $x\in\R^d$ satisfying $\mathcal{L} \cup \{x^T c \geq t\}$ or ``NOT FOUND.''
 \end{enumerate}
\end{lemma}
\begin{proof}
 The idea is to run the ellipsoid method for LP feasibility, but we make sure to terminate the process early after $O(d^2 \log(\tau''/\epsilon))$ iterations. At this point, if the center of the final ellipsoid is feasible, we output the center; otherwise, we output ``NOT FOUND.'' It is evident that rules (2) and (3) will be satisfied by the algorithm.
\shortversion{
\begin{figure}
\centering
 \begin{subfigure}[t]{0.2\textwidth}
  \centering
  \includegraphics[height=1.2in]{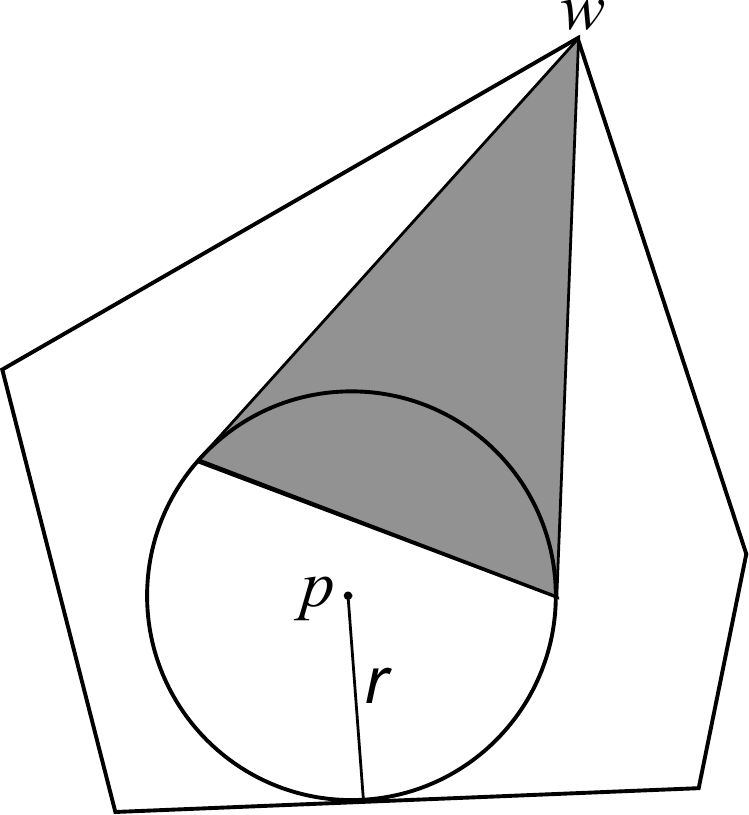}
  \caption{The shaded volume shows the cone formed by joining $w$ to $\ball(p,r)$. It is contained in
the Voronoi cell of $p$.}
  \label{fig:cone}
 \end{subfigure}
\quad
 \begin{subfigure}[t]{0.2\textwidth}
  \centering
  \includegraphics[height=1.2in]{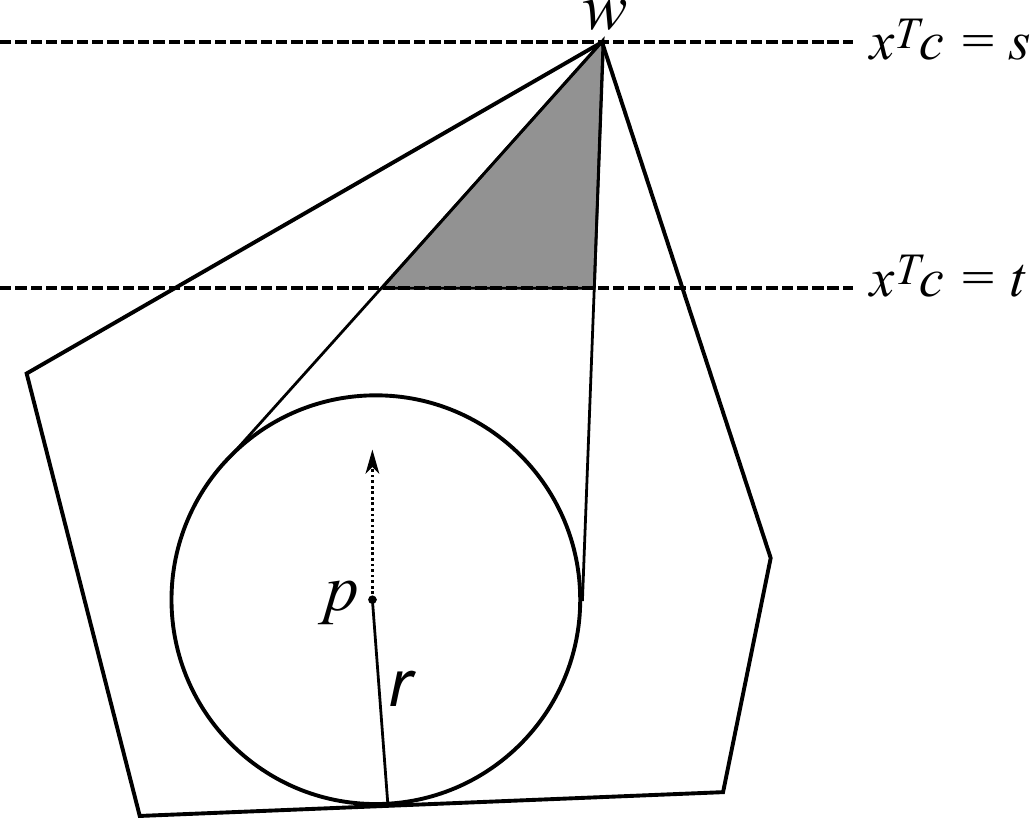}
  \caption{The shaded volume shows the cone formed by taking the intersection of the cone in (a)
with the half space $\{x\in\R^d: x^T c\geq t\}$. The dotted arrow indicates the direction given by
$c$.}
  \label{fig:cone_trunc}
 \end{subfigure}
 \caption{Bounding the volume of the feasible region during the ellipsoid method}
\end{figure}
}

 Let us show that (1) is also satisfied. Suppose $t < s\left(1-\frac{\epsilon}{2}\right)$. Then, the
corresponding LP is feasible. We shall now bound the volume of the feasible region $\mathcal{R}$
from below. Recall that the ball $\ball(p, r)$ of radius $r$ centered at $p$ is entirely contained in
$\vor(p)$. Thus, the $d$-dimensional cone joining apex $w$ to $\ball(p,r)$ is contained in $\vor(p)$
(see Figure~\ref{fig:cone}). Hence, the intersection of this cone with $\{x\in\R^d: x^T c \geq t\}$
(which is itself a cone with an ellipsoid base) is contained inside $\mathcal{R}$ (see
Figure~\ref{fig:cone_trunc}). Using elementary geometry, we find that the volume of this
intersection is at least
 $$
  \frac{1}{d} C_{d-1} \frac{r^{d-1} (s-t)^d}{(s^2 - r^2)^{\frac{d-1}{2}}},
 $$
where $C_j$ is the volume of a $j$-dimensional unit ball 
\fullversion{(see Appendix~\ref{sec:ellipsoid} for a detailed derivation).}
\shortversion{(a detailed derivation can be found in the full version~\cite{ARXIV_VERSION}).} 
Using the assumptions $t < s\left(1-\frac{\epsilon}{2}\right)$ and $s \geq
\frac{1}{1-\epsilon}r$, we see that the above quantity is defined and is at least
$$
\frac{1}{d} C_{d-1} \left(\frac{\epsilon}{2}\right)^d \frac{r^{d-1} s^d}{(s^2 -
r^2)^\frac{d-1}{2}}.
$$
The above quantity is minimized for $s = \sqrt{d} r$. Thus, it follows that the volume of
the feasible region is bounded from below by
$$
\frac{1}{\sqrt{d}} C_{d-1} \left(\frac{\epsilon r}{2}\right)^d
\left(1+\frac{1}{d-1}\right)^\frac{d-1}{2} \geq \sqrt{\frac{2}{d}} C_{d-1} \left(\frac{\epsilon
r}{2}\right)^d.
$$
Moreover, by Theorem~\ref{thm:quality_invariant}, we know that $\ball(p, r\tau'')$ contains $\vor(p)$
and, therefore, the feasible
region as well. Hence, by a standard analysis of the ellipsoid method (see \cite{GroetschelLovaszSchrijver1993}), the number of
iterations required to correctly output a feasible point is $O\left(d
\log\left(\frac{vol(\ball(p,r\tau''))}{vol(\mathcal{R})}\right)\right)$, which we compute as
\begin{eqnarray*}
 O\left(d \log\left( \frac{C_d (r\tau'')^d}{\sqrt{\frac{2}{d}} C_{d-1} (\epsilon
r/2)^d}\right)\right) = O\left(d^2 \log\left(\frac{\tau''}{\epsilon}\right)\right).
\end{eqnarray*}
Terminating the ellipsoid method after $O\left(d^2 \log\left(\frac{\tau''}{\epsilon}\right)\right)$ iterations  therefore suffices to obtain the desired output, hence establishing condition (1).

Finally, since performing each iteration of the ellipsoid method can be done in $2^{O(d)}$ time, the total running time is $2^{O(d)} O\left(\log\left(\frac{\tau''}{\epsilon}\right)\right)$, as desired.
\end{proof}
Now, we are ready for the proof of Theorem~\ref{thm:ellipsoid}.
\begin{proof}[\fullversion{Proof }of Theorem~\ref{thm:ellipsoid}]
 We show how to obtain $z$. Let $s = w^T c$. The trick is to run the feasibility algorithm of Lemma~\ref{lem:threshellipsoid} on the linear program with constraint set $\mathcal{L} \cup \{x^T c \geq t\}$ for different values of $t$. Let $L = \left\lceil \frac{2(\tau''-1)}{\epsilon} \right\rceil$.	Then, let $t_0, t_1, t_2, \dots, t_L \in [r, r\tau'']$ be defined by $t_i = r + \frac{r\tau''-r}{L}i$. In particular, we do a binary search to find $i\in \{0,1,\dots,L\}$ such that the algorithm of Lemma~\ref{lem:threshellipsoid} returns a solution for $t=t_i$ but reports ``NOT FOUND'' for $t=t_{i+1}$. Then, we simply output the $z$ corresponding to the obtained solution for $t=t_i$. If such an $i$ does not exist (i.e., the algorithm does not return a feasible solution for any $i$), we simply output $z = rc$. Note that the total runtime of this overall procedure is $2^{O(d)} O\left(\log^2\left(\frac{\tau''}{\epsilon}\right)\right)$.
 
 It remains to show that $z^T c \geq (1-\epsilon) s$.
 \begin{enumerate}
  \item If $s < \frac{1}{1-\epsilon}r$, then note that any $z$ that is outputted satisfies $z^T c \geq r > (1-\epsilon)s$.
  \item Suppose $s\geq\frac{1}{1-\epsilon}r$. Let $i$ be the index returned by the binary search for which $t=t_i$ returns a solution but $t=t_{i+1}$ returns ``NOT FOUND.'' Clearly, we must have that $t_{i+1} \geq \left(1-\frac{\epsilon}{2}\right)s$. Thus, we get the desired bound on $t_i$ as follows.
  \begin{align*}
    t_i &= t_{i+1} - \frac{r\tau'' - r}{L}
    \geq \left(1-\frac{\epsilon}{2}\right)s - \frac{\epsilon r}{2}\\
    &\geq \left(1-\frac{\epsilon}{2}\right)r - \frac{\epsilon r}{2}
    = (1-\epsilon)r.
  \end{align*}
 \end{enumerate}
%
%
%
%
\end{proof}
Now, the farthest corner in the direction of an arbitrary unit vector $c$ is, in general, not the farthest corner of the entire Voronoi cell. Hence, we need the following lemma.
\begin{lemma}\label{lem:theta_approx}
  Let $p$ be a point in the mesh with $q$ as the farthest corner of its Voronoi cell. 
  Let $0<\theta<\pi/2$ and let $c\in S^{d-1}$ such that the angle between $c$ and $q-p$ is $< \theta$. 
  Then, applying the underlying algorithm of Theorem~\ref{thm:ellipsoid} with this choice of $c$ produces a point $z$ inside $\vor(p)$ for which $ \dist(p,z) \geq (1-\epsilon)(\cos\theta) \dist(p,q)$.
\end{lemma}
\begin{proof}
As before, assume without loss of generality that $p$ is located at the origin. Let $w$ be the farthest corner of $\vor(p)$ in the direction of $c$. The algorithm of Theorem~\ref{thm:ellipsoid} produces a point $z$ inside $\vor(p)$ such that $z^T c \geq (1-\epsilon)w^T c$. The desired claim follows since $w^T c \geq q^T c \geq \dist(p,q)\cos\theta$.
\end{proof}
As a corollary, we obtain the following:
\begin{corollary}\label{cor:approxfarthest}
 If $q$ is the farthest corner of $\vor(p)$, then repeating the algorithm of Theorem~\ref{thm:ellipsoid} for each $c\in\etaCAGE$ and taking the point $z$ that is farthest from $p$ among all outputted points gives us $\dist(p,z) \geq (1-\epsilon)\left(1-\frac{\eta^2}{2}\right) \dist(p,q)$.
\end{corollary}
\begin{proof}
 Consider the point $z' = (z-p)/ \dist(p,z)$. Since $z'\in S^{d-1}$, there exists $c\in \etaCAGE$ such that $\dist(z', c) \leq \eta$. Since $z'$ and $c$ are both unit vectors, elementary geometry implies that the angle $\theta$ between $z'$ and $c$ satisfies $\cos\theta = 1-\frac{ \dist(z',c)^2}{2} \geq 1-\frac{\eta^2}{2}$. Thus, Lemma \ref{lem:theta_approx} implies the desired claim.
\end{proof}

\subsection{The Greedy Walk has Constant Length} 
\label{sec:greedy_walk}

  The greedy walk is named as such because it uses the predecessor pairing from the greedy permutation to start the walk.
  This is in contrast to previous walking-based point location schemes that use randomization.
  Also, it walks through the Voronoi diagram rather than the Delaunay simplices and therefore needs no hyperplane tests.
  In this section, we show that the greedy walk only takes a constant number of steps before terminating.

\begin{lemma}\label{lem:greedy_walk}
  For an input point $p$, the greedy walk from $\phi(p)$ to $p$ only takes $2^{O(d)}$ steps.
\end{lemma}

\begin{proof}
  Let $M$ denote the current set of inserted points at the time we insert $p$.
  Let $S$ be the line segment $\overline{p\phi(p)}$, and let $\ell$ denote its length.
  Let $m$ be the midpoint of $S$.
  The proof will be by a volume packing argument.
  It will suffice to show that any Voronoi cell $\vor_M(q)$ intersecting $S$ has an inradius in the range $[\gamma\ell, \ell]$ for some constant $\gamma=1/(4K(1+\tau))$ and also that $ \dist(q,m)\le 2\ell$.
  Thus, the inballs are contained in $\ball(m, 3\ell)$ and their number is at most $(3/\gamma)^d = 2^{O(d)}$.
  
  First we observe that for any $x\in S$ the feature size $\lfs_{P_i}(x)$ is at least $\ell/2$.
  Otherwise, there must be a point $y\in P_{i-1}$ that is closer than $\ell$ to either $p_i$ or $\phi(p_i)$, but this would violate the greedy ordering.
  
  Let $q$ be any point of $M$ such that $\vor_M(q)$ intersects $S$.
  So, some $x\in S$ is closer to $q$ than to $\phi(p_i)$ and by the triangle inequality, $ \dist(q, \phi(p_i))\le 2\ell$.
  The inradius $r(q)$ of $\vor_M(q)$ is at most half the distance from $q$ to $\phi(p_i)$ and thus is at most $\ell$.
  To lower bound $r(q)$ we use Theorem~\ref{thm:fs_invariant} and the quality of $M$ as follows.
  \begin{align*}
    r(q) 
      &= \frac{1}{2}\lfs_{M}(q)
      \ge \frac{1}{2K}\lfs_{P_i}(q)
      \ge \frac{1}{2K}\lfs_{P_i}(x) -  \dist(x,q)\\
      &\ge \frac{1}{4K}\ell - \tau r(q).
  \end{align*}
  Thus we get that $r(q)\ge \ell/(4K(1+\tau)) = \gamma\ell$.
  The result follows from the aforementioned volume packing argument.
\end{proof}


\subsection{Running Time} 
\label{sec:running_time}

  
  \begin{theorem}\label{thm:running_time}
    The running time of our algorithm is \shortversion{\\}$2^{O(d)}(n\log n + m)$.
  \end{theorem}
  \begin{proof}
    The expected running time of the greedy permutation algorithm is $2^{O(d)}n\log n$.
    Each step of the algorithm adds a new point so it suffices to show that each insertion takes only $2^{O(d)}$ time.
    Lemmas~\ref{lem:greedy_walk} and~\ref{lem:nrb_bound} imply that the point location work is only $2^{O(d)}$ per input point.
    Moreover, the degree bound in Lemma~\ref{lem:nrb_bound} also implies that there are only $2^{O(d)}$ vertices to update with any insertion.
    Each takes only $2^{O(d)}$ time to locate the new neighbors and recompute the approximate farthest corners, by Theorem~\ref{thm:ellipsoid}.
    Thus, the total running time after preprocessing is $2^{O(d)}m$.
  \end{proof}

  \section{Conclusion} 
\label{sec:conclusion}

  We have presented several new ideas for mesh generation.
  The most relevant are the use of linear programming to find large empty regions in the mesh without resorting to enumerating the simplices of the Delaunay triangulation and the use of greedy permutations as a preprocess for walking-based point location.
  The latter idea alone could replace the rather complicated range net methods of previous work to get running times independent of the spread~\cite{miller11beating}.


  \bibliographystyle{abbrv}
  \bibliography{../bibliography}

\begin{thebibliography}{10}

\bibitem{cgalMesh3D}
P.~Alliez, L.~Rineau, S.~Tayeb, J.~Tournois, and M.~Yvinec.
\newblock {3D} mesh generation.
\newblock In {\em {CGAL} User and Reference Manual}. {CGAL Editorial Board},
  {4.1} edition, 2012.
\newblock
  \url{http://www.cgal.org/Manual/4.1/doc_html/cgal_manual/packages.html#Pkg:M%
esh_3}.

\bibitem{boissonnat09incremental}
J.-D. Boissonnat, O.~Devillers, and S.~Hornus.
\newblock {Incremental construction of the Delaunay graph in medium dimension}.
\newblock In {\em {Annual Symposium on Computational Geometry}}, pages
  208--216, 2009.

\bibitem{cheng12delaunay}
S.-W. Cheng, T.~K. Dey, and J.~R. Shewchuk.
\newblock {\em Delaunay Mesh Generation}.
\newblock CRC Press, 2012.

\bibitem{chew89guaranteed}
L.~P. Chew.
\newblock Guaranteed-quality triangular meshes.
\newblock Technical Report TR-89-983, Department of Computer Science, Cornell
  University, 1989.

\bibitem{clarkson06building}
K.~L. Clarkson.
\newblock Building triangulations using epsilon-nets.
\newblock In {\em STOC: ACM Symposium on Theory of Computing}, pages 326--335,
  2006.

\bibitem{green78computing}
P.~J. Green and R.~Sibson.
\newblock Computing {D}irichlet tessellations in the plane.
\newblock {\em Computer Journal}, 21(2):168--173, 1978.

\bibitem{GroetschelLovaszSchrijver1993}
M.~Gr{\"o}tschel, L.~Lov{\'a}sz, and A.~Schrijver.
\newblock {\em {Geometric Algorithms and Combinatorial Optimization}}, volume~2
  of {\em Algorithms and Combinatorics}.
\newblock Springer, second corrected edition edition, 1993.

\bibitem{har-peled06fast}
S.~Har-Peled and M.~Mendel.
\newblock Fast construction of nets in low dimensional metrics, and their
  applications.
\newblock {\em SIAM Journal on Computing}, 35(5):1148--1184, 2006.

\bibitem{hudson06sparse}
B.~Hudson, G.~Miller, and T.~Phillips.
\newblock {S}parse {V}oronoi {R}efinement.
\newblock In {\em Proceedings of the 15th International Meshing Roundtable},
  pages 339--356, Birmingham, Alabama, 2006.
\newblock Long version available as Carnegie Mellon University Technical Report
  CMU-CS-06-132.

\bibitem{hudson10topological}
B.~Hudson, G.~L. Miller, S.~Y. Oudot, and D.~R. Sheehy.
\newblock Topological inference via meshing.
\newblock In {\em SOCG}, pages 277--286, 2010.

\bibitem{matousek02lectures}
J.~Matou\v{s}ek.
\newblock {\em Lectures on Discrete Geometry}.
\newblock Springer-Verlag, 2002.

\bibitem{miller11beating}
G.~L. Miller, T.~Phillips, and D.~R. Sheehy.
\newblock Beating the spread: Time-optimal point meshing.
\newblock In {\em SOCG}, pages 321--330, 2011.

\bibitem{miller99radius}
G.~L. Miller, D.~Talmor, S.-H. Teng, and N.~Walkington.
\newblock On the radius-edge condition in the control volume method.
\newblock {\em SIAM J. on Numerical Analysis}, 36(6):1690--1708, 1999.

\bibitem{mucke96fast}
E.~P. M{\"u}cke, I.~Saias, and B.~Zhu.
\newblock Fast randomized point location without preprocessing in two-and
  three-dimensional {Delaunay} triangulations.
\newblock {\em Comput. Geom}, 12(1-2):63--83, 1999.

\bibitem{ruppert95delaunay}
J.~Ruppert.
\newblock A {D}elaunay refinement algorithm for quality $2$-dimensional mesh
  generation.
\newblock {\em J. Algorithms}, 18(3):548--585, 1995.

\bibitem{sheehy11mesh}
D.~R. Sheehy.
\newblock {\em Mesh Generation and Geometric Persistent Homology}.
\newblock PhD thesis, Carnegie Mellon University, 2011.

\bibitem{sheehy12new}
D.~R. Sheehy.
\newblock {New Bounds on the Size of Optimal Meshes}.
\newblock {\em Computer Graphics Forum}, 31(5):1627--1635, 2012.

\bibitem{shewchuk96triangle}
J.~R. Shewchuk.
\newblock Triangle: Engineering a 2{D} quality mesh generator and {D}elaunay
  triangulator.
\newblock In {\em Applied Computational Geometry}, volume 1148 of {\em Lecture
  Notes in Computer Science}, pages 203--222, 1996.

\bibitem{siTetGen}
H.~Si.
\newblock Tet{G}en: A quality tetrahedral mesh generator and a {3D} {D}elaunay
  triangulator.
\newblock \url{http://tetgen.org/}, January 2011.

\end{thebibliography}
  
  \renewcommand{\thesection}{\Alph{section}}
  \setcounter{section}{0}
  
  \section{Pseudocode}
\label{sec:pseudocode}

We present the pseudocode for the algorithm with the goal of giving the most concise description that both allows a rigorous analysis and could guide an implementation.
The blocks of code are whitespace delimited, i.e.\ indentation indicates scope.
 
We assume a graph data structure supporting dynamic insertions and deletions using the methods \InsertEdge, \DeleteEdge, \InsertVertex, and \DeleteVertex.
The vertices are the geometric points.
The approximate Delaunay neighbors of a vertex $v$ can be accessed from $\Nbr(v)$.

The main routine for constructing a $\tau$-well-spaced set of points, given an input set $P$ and quality parameter $\tau$, is as follows.
We indicate $\tau$ as an input parameter, but we assume that it is available to all of the methods without explicitly passing it.

\algoline
\WellSpacedPoints($P$: points, $\tau$: quality parameter) \\
\tab  $(P,\phi) \leftarrow \ApproximateGreedyPermutation(P)$\\
\tab  $\InsertVertex(p_1)$\\
\tab  $\NewLayer(p_2,p_1)$\\
\tab  \For $i=3$ \To $|P|$\\
\tab\tab    $\Insert(p_i)$\\
\tab\tab    $\Refine()$

\algoline
The following is a routine for considers a new point $p$ for insertion.

\algoline
\Insert($p$: point)\\
\tab  $q \gets \FindNN(p,\phi(p))$\\
\tab  \If $q$ is a Steiner point \Then $\Snap(p,q)$\\
\tab  \ElseIf $\dist(p,q) < \frac{1}{K} r(q)$ \Then $\NewLayer(p,q)$\\
\tab  \Else $\RegularInsert(p,q)$

\algoline
\RegularInsert inserts a new point $p$ that lies in the Voronoi cell of an already-inserted point $q$. \Snap snaps a Steiner point $q$ to an input point $p$ which is then inserted. \NewLayer adds a new layer to the hierarchy, places point $q$ inside, and then inserts the input point $p$.

\algoline
\RegularInsert($p$: point, $q$: point)\\
\tab $\InsertVertex(p)$\\
\tab  $r(p)\gets \dist(p,q)/2$\\
\tab  Do a graph search starting from $q$:\\
\tab\tab    At a vertex $v$, \If $\dist(v,p)\le 2\min\{R(v), \tau'' r(p)\}$ \\
\tab\tab    \Then \InsertEdge$(p,v)$ and search the neighbors of $v$\\
\tab  $\UpdateAspect(p)$\\
\tab  $\PruneEdges(p)$\\
\tab  \ForEach $v\in \Nbr(p)$\\
\tab\tab    $\UpdateAspect(v)$\\
\tab\tab    $\PruneEdges(v)$
    
\algoline
\Snap($p$: point, $q$: point)\\
\tab  $q'\gets \argmin_{p'\in\Nbr(q)}\dist(p,p')$\\
\tab  $\Delete(q)$\\
\tab  $\RegularInsert(p,q')$

\algoline
\NewLayer($p$: point, $q$: point)\\
\tab  Create a new copy of $\etaCAGE$ with center $q$ and radius $2\dist(p,q)$\\
\tab  \ForEach pair of vertices $(u,v)$\\
\tab\tab $\InsertEdge(u,v)$\\
\tab  $\RegularInsert(p,q)$\\
\algoline
Next is a procedure to remove a point $p$ that has previously been inserted. The routine is used during snapping, when a Steiner point must be removed to make way for a nearby input point.

\algoline
\Delete($p$: point)\\
\tab  $N\gets \Nbr(p)$\\
\tab  $\DeleteVertex(p)$\\
\tab  \ForEach $u,v\in N$\\ 
\tab\tab    $\InsertEdge(u,v)$\\
\tab  \ForEach $v\in N$\\
\tab\tab    $\UpdateAspect(v)$\\
\tab  \ForEach $v\in N$\\
\tab\tab    $\PruneEdges(v)$\\
\algoline
The \Refine procedure repeatedly adds Steiner points at the approximate farthest corners of any cells whose aspect ratio is too large.

\algoline
\Refine()\\
\tab  \While $\Qrefine$ is nonempty\\
\tab\tab    $q\gets $ pop($\Qrefine$)\\
\tab\tab    \If $R(q)/r(q) > \tau$ \Then $\RegularInsert(\ApproxFarCorner(q), q)$\\
\algoline
The next two methods update the data structure.
\UpdateAspect computes the inradius and (approximate) outradius of a vertex, and it puts that vertex on the refinement queue if its aspect ratio is greater than $\tau$. 

\algoline
\UpdateAspect($p$: point)\\
\tab  $a\gets R(p)/r(p)$\\
\tab  $r(p)\gets \frac{1}{2} \min_{q\in \Nbr(p)} \dist(p,q)$\\
\tab  $x \gets \ApproxFarCorner(p)$\\
\tab  $R(p)\gets \dist(x,p)$\\
\tab  \If $R(p)/r(p)>\tau\ge a$ \Then push $p$ to $\Qrefine$
  
\algoline
\PruneEdges($p$: point)\\
\tab  \ForEach $q\in \Nbr(p)$\\
\tab\tab    \If $\dist(p,q) > 2 \min\{R(p), R(q)\}$ \Then $\DeleteEdge(p,q)$

\algoline
The next procedure does a greedy walk along $\overline{p'p}$, starting at $p'$, and returns the point whose Voronoi cell contains $p$. It should be noted that the notation $\mathcal{H}(u,v)$ denotes the hyperplane that bisects $\overline{uv}$.

\algoline
\newcommand{\vnext}{v_{\mathrm{next}}}
\newcommand{\vcurrent}{v_{\mathrm{current}}}
\FindNN($p$: point, $p'$: point)\\
\tab $\vcurrent \leftarrow p'$\\
\tab $\vnext \leftarrow p'$\\
\tab $w \leftarrow p'$\\
\tab \pseudocode{while} $w$ lies on $\overline{p'p}$\\
\tab\tab $s\leftarrow\infty$\\
\tab\tab \ForEach $q'\neq \vcurrent$ in $\Nbr(\vnext)$\\
\tab\tab\tab \If $\overline{wp}$ intersects $\mathcal{H}(\vnext, q')$ and $\langle p-w, q'-\vnext\rangle \geq 0$ \Then\\
\tab\tab\tab\tab $t'\leftarrow \overline{wp} \cap \mathcal{H}(\vnext, q')$\\
\tab\tab\tab\tab \If $ \dist(w,t') < s$ \Then\\
\tab\tab\tab\tab\tab $t\leftarrow t'$\\
\tab\tab\tab\tab\tab $q\leftarrow q'$\\
\tab\tab\tab\tab\tab $s\leftarrow  \dist(w,t')$\\
\tab\tab \If $s\neq\infty$\\
\tab\tab\tab $w\leftarrow t$\\
\tab\tab\tab $\vcurrent\leftarrow \vnext$\\
\tab\tab\tab $\vnext\leftarrow q$\\
\tab\tab \Else\\
\tab\tab\tab \Return $\vcurrent$

\algoline
The following routine takes in a point $p$ and locates an approximate farthest corner of $\vor(p)$, i.e., returns a point $z$ inside $\vor(P)$ such that $ \dist(p,z)$ is within a factor $(1-\epsilon)\left(1-\frac{\eta^2}{2}\right)$ of the distance from $p$ to the true farthest corner:

\algoline
\ApproxFarCorner($p$: point)\\ 
\tab $r\gets r(p)$\\
\tab $z\leftarrow p$\\
\tab \ForEach $c \in \etaCAGE$ \\
\tab\tab $w\gets rc$\\
\tab\tab $L\gets \left\lceil\frac{2(\tau''-1)}{\epsilon}\right\rceil$\\
\tab\tab $j\gets 0$\\
\tab\tab $k\gets L+1$\\
\tab\tab \pseudocode{while} $k-j > 1$\\
\tab\tab\tab $i\gets \left\lceil \frac{j+k}{2} \right\rceil$\\
\tab\tab\tab $t\gets r + \frac{r\tau'' - r}{L}i$.\\
\tab\tab\tab Consider the LP feasibility problem:\\
\tab\tab\tab\tab $x\in\R^d$ \pseudocode{subject to:} $\forall q \in \Nbr(p)$, $(q-p)^T x \leq \frac{(q-p)^T (q+p)}{2}$,\\
\tab\tab\tab\tab \pseudocode{\ \ \ \ \ \ \ \ \ \ \ \ \ \ \ \ \ \ \ \ \ \ \,} $c^T (x-p) \geq t$\\
\tab\tab\tab Run $O(d^2 \log(\tau''/\epsilon))$ iterations of the ellipsoid method\\
\tab\tab\tab \pseudocode{if} a feasible solution $x$ is produced \pseudocode{then}\\
\tab\tab\tab\tab $j\gets i$\\
\tab\tab\tab\tab $w\gets x$\\
\tab\tab\tab \pseudocode{else}\\
\tab\tab\tab\tab $k\gets i$\\
\tab\tab \pseudocode{if} $ \dist(p,w) > \dist(p,z)$ \pseudocode{then} $z\leftarrow w$\\
\tab \pseudocode{return} z\\
\algoline
Finally, \ApproximateGreedyPermutation is a blackbox that takes in a set $P$ of points and returns a permutation $p_1, p_2, \dots, p_{|P|}$ of the points along with the predecessor pairing $\phi$.
  \section{Ellipsoid Volume Calculation}
\label{sec:ellipsoid}

In the proof of Lemma~\ref{lem:threshellipsoid}, we state (without proof) that the volume of the
intersection of the cone in Figure~\ref{fig:cone_trunc} is at least
$$
\frac{1}{d} C_{d-1} \frac{r^{d-1} (s-t)^d}{(s^2 - r^2)^{\frac{d-1}{2}}}.
$$
We shall prove that claim here.

Without loss of generality, suppose the apex $w$ of the cone is located at the origin
and that the coordinates of $p$ are $(a,\underbrace{0,0,\dots,0}_{d-2},s)$. Then,
we wish to calculate the volume of the region given by the intersection of the cone with $\{(x_1,
x_2, \dots, x_d)\in\R^d: 0\leq x_d\leq s-t\}$. Note that this region itself is a cone with an
ellipsoid base and height $s-t$.

Let us determine the ellipsoid base. Note that the boundary of the ellipsoid is given by the
intersection of the boundary of the cone with the hyperplane $x_d = s-t$. Note that the angle formed by
the side of the cone with the axis is $\theta$, where $\cos\theta =
\sqrt{\frac{a^2+s^2-r^2}{a^2+s^2}}$. Hence, the equation of the cone is
\begin{eqnarray*}
 \left\langle \left(\frac{a}{\sqrt{a^2+s^2}}, \underbrace{0,0,\dots,0}_{d-2}, \frac{s}{\sqrt{a^2+s^2}}\right), (x_1,
x_2,\dots, x_d)\right\rangle = \sqrt{x_1^2 + x_2^2 + \dots + x_d^2} \cos\theta,
\end{eqnarray*}
or,
\begin{eqnarray*}
 \left(\frac{ax_1}{\sqrt{a^2+s^2}} + \frac{sx_d}{\sqrt{a^2+s^2}}\right)^2 = (x_1^2 + x_2^2 + \cdots
+ x_d^2)\cos^2\theta = (x_1^2 + \cdots + x_d^2)\frac{a^2+s^2-r^2}{a^2+s^2},
\end{eqnarray*}
where $\langle\, ,\,\rangle$ denotes the standard inner product on $\R^d$. To compute the intersection with the hyperplane $x_d = s-t$, we plug in $x_d = s-t$ into the above equation to
obtain
\begin{eqnarray*}
 \dfrac{\left(x_1 -
\frac{as(s-t)}{s^2-r^2}\right)^2}{\left(\frac{r(s-t)\sqrt{a^2+s^2-r^2}}{s^2-r^2}\right)^2} +
\dfrac{1}{\left(\frac{r(s-t)}{\sqrt{s^2-r^2}}\right)^2} (x_2^2 + x_3^2 + \cdots + x_{d-1}^2) &=& 1.
\end{eqnarray*}
Hence, the ellipsoid base of the cone whose volume we desire has semiaxes with lengths
$$
 \frac{r(s-t)\sqrt{a^2+s^2-r^2}}{s^2-r^2}, \underbrace{\frac{r(s-t)}{\sqrt{s^2-r^2}},
\frac{r(s-t)}{\sqrt{s^2-r^2}}, \dots, \frac{r(s-t)}{\sqrt{s^2-r^2}}}_{d-2}.
$$
Since the height of the cone is $s-t$, it follows that the volume is
\begin{eqnarray*}
 \frac{1}{d}(s-t)C_{d-1}
\left(\frac{r(s-t)\sqrt{a^2+s^2-r^2}}{s^2-r^2}\right)\left(\frac{r(s-t)}{\sqrt{s^2-r^2}}\right)^{d-2
} \geq \frac{1}{d} C_{d-1} \frac{r^{d-1} (s-t)^d}{(s^2-r^2)^{\frac{d-1}{2}}}.
\end{eqnarray*}

\end{document}